\documentclass{article}
\pdfoutput=1

\usepackage{microtype, url, hyperref, listings, amsmath}
\usepackage{amssymb, amsthm}
\usepackage[shortlabels]{enumitem}
\usepackage{xspace}
\usepackage[a4paper]{geometry}

\usepackage{xcolor}
\usepackage[linesnumbered,ruled,vlined]{algorithm2e}

\SetCommentSty{mycommfont}
\SetKwInput{KwData}{Input}

\title{Modulo Counting on Words and Trees%
\footnote{Supported by the Polish National Science
Centre grant No.~2016/21/B/ST6/01444.}}

\author{Bartosz Bednarczyk\\
  {\small Computer Science Department, ENS Paris-Saclay Cachan, France; and}\\
  {\small Institute of Computer Science, University of Wroc\l aw Wroc\l aw, Poland}\\
  {\small \texttt{bartosz.bednarczyk@ens-paris-saclay.fr}}
  \and Witold Charatonik \\
  {\small Institute of Computer Science, University of Wroc\l aw Wroc\l aw, Poland}\\
  {\small \texttt{Witold.Charatonik@cs.uni.wroc.pl}}}

\date{}

\newtheorem{theorem}{Theorem}
\newtheorem{lemma}[theorem]{Lemma}
\newtheorem{corollary}[theorem]{Corollary}

\newtheorem{definition}[theorem]{Definition}

\begin{document}
  \maketitle
%%% OUR MACROS

\newcommand{\DDD}{\mbox{\large \boldmath $\delta$}}
\newcommand{\cupdot}{\mathbin{\mathaccent\cdot\cup}}
\renewcommand{\phi}{\varphi} % Nicer-looking phi
\newcommand{\AAA}{\mbox{\large \boldmath $\alpha$}}
\newcommand{\BBB}{\mbox{\large \boldmath $\beta$}}
\newcommand{\EQ}{\ensuremath{{\mathcal EQ}}}
\newcommand{\Sat}{\ensuremath{\textit{Sat}}}
\newcommand{\FinSat}{\ensuremath{\textit{FinSat}}}
\newcommand{\FO}{\mbox{\rm FO}}
\newcommand{\MSO}{\mbox{\rm MSO}}
\newcommand{\FOt}{\mbox{$\mathrm{FO}^2$}}
\newcommand{\FOth}{\mbox{$\mbox{\rm FO}^3$}}
\newcommand{\Ct}{\mbox{$\mbox{\rm C}^2$}}

% Complexity classes

\newcommand{\nlogspace}{\textsc{NLogSpace}}
\newcommand{\np}{\textsc{NPTime}}
\newcommand{\ptime}{\textsc{PTime}}
\newcommand{\pspace}{\textsc{PSpace}}
\newcommand{\exptime}{\textsc{ExpTime}}
\newcommand{\expspace}{\textsc{ExpSpace}\xspace}
\newcommand{\nexptime}{\textsc{NExpTime}}
\newcommand{\twoexptime}{2\textsc{-ExpTime}\xspace}
\newcommand{\toexptime}{2\textsc{-NExpTime}}
\newcommand{\twoexpspace}{2\textsc{-ExpSpace}}
\newcommand{\apspace}{\textsc{APSpace}}
\newcommand{\aptime}{\textsc{APTime}}
\newcommand{\atwoexptime}{\textsc{A}(2\textsc{-ExpTime})\xspace}
\newcommand{\aexpspace}{\textsc{AExpSpace}\xspace}
\newcommand{\aspace}{\textsc{ASpace}}
\newcommand{\atime}{\textsc{ATime}}
\newcommand{\aexptime}{\textsc{AExpTime}}
\newcommand{\dtime}{\textsc{DTime}}
\newcommand{\nptime}{\textsc{NPTime}}

% Other useful macros

\newcommand{\set}[1]{\{#1\}}
\newcommand{\md}[2][] {{\lfloor#2\rfloor_{#1}}}
\newcommand{\sizeOf}[1]{| #1 |}
\newcommand{\str}[1]{{\mathfrak{#1}}}
\newcommand{\restr}{\!\!\restriction\!\!}
\newcommand{\tuple}[1]{\langle#1\rangle}

\newcommand{\N}{{\mathbb N}}
\newcommand{\Z}{{\mathbb Z}} 

% words

\newcommand{\wordsucc}{\mathit{succ}}

\newcommand{\tl}{\theta_{\ll}}
\newcommand{\tg}{\theta_{\gg}}
\newcommand{\tprev}{\theta_{-1}}
\newcommand{\tnext}{\theta_{+1}}
\newcommand{\teq}{\theta_{=}}

\newcommand{\FOTwoModWords}{\mbox{$\mbox{\rm FO}_{\rm MOD}^2[\leq, \wordsucc]$}}
\newcommand{\FOThreeModWords}{\mbox{$\mbox{\rm FO}_{\rm MOD}^3[\leq, \wordsucc]$}}
\newcommand{\FOModWords}{\mbox{$\mbox{\rm FO}_{\rm MOD}[\leq, \wordsucc]$}}

% commands for the version with ordered trees

\newcommand{\tfree}{\theta_{\not\sim}}
\newcommand{\succv}{{\downarrow}}
\newcommand{\lessv}{{\downarrow^{\scriptscriptstyle +}}}
\newcommand{\succh}{{\rightarrow}}
\newcommand{\lessh}{{\rightarrow^{\scriptscriptstyle +}}}
\newcommand{\succlessv}{{\downarrow \downarrow^+}}
\newcommand{\predgreatv}{{\uparrow \uparrow^+}}

\newcommand{\tsuccv}{\theta_{\downarrow}}
\newcommand{\tprecv}{\theta_{\uparrow}}
\newcommand{\tlessv}{\theta_{\downarrow \downarrow^+}}
\newcommand{\tgreatv}{\theta_{\uparrow \uparrow^+}}
\newcommand{\tsucch}{\theta_{\rightarrow}}
\newcommand{\tprech}{\theta_{\leftarrow}}
\newcommand{\tlessh}{\theta_{\rightrightarrows^+}}
\newcommand{\tgreath}{\theta_{\leftleftarrows^+}}

\newcommand{\FOtall}{\mbox{$\FOt[\succv, \lessv, \succh, \lessh, \tau_{com}]$}}
\newcommand{\Ctall}{\mbox{$\Ct[\succv, \lessv, \succh, \lessh, \tau_{com}]$}}
\newcommand{\CTwoFull}{\mbox{$\Ct[\succv, \lessv, \succh, \lessh]$}}

\newcommand{\cutout}[1]{}

\newcommand{\fw}{{\not\sim}}
\newcommand{\Nn}{\mathbb{N}}
\newcommand{\tree}{\str{T}}
\newcommand{\Ll}{L}
\newcommand{\trs}{I}
\newcommand{\hv}{\hat{v}}
\newcommand{\greatv}{{\uparrow^{\scriptscriptstyle +}}}
\newcommand{\greateqv}{{\uparrow^{\scriptscriptstyle *}}}
\newcommand{\lesseqv}{{\downarrow_{\scriptscriptstyle *}}}

%% typy

\newcommand{\tp}[2]{{\rm tp}^{\str{#1}}({#2})}
\newcommand{\ftp}[3]{{\rm ftp}^{\str{#1}}_{#2}({#3})}
\newcommand{\rftpa}[2]{{\rm rftp}_{{#1}}({#2})}
\newcommand{\rftpv}[3]{{\rm rftp}^{\str{#1}}_{#2}({#3})}
\newcommand{\hftp}[3]{{\rm hftp}^{\str{#1}}_{#2}({#3})}
\newcommand{\hshift}{\hspace*{20pt}}
\newcommand{\hshiftshift}{\hspace*{30pt}}

\newcommand{\ftpA}{\overline{\alpha}}
\newcommand{\ftpB}{\overline{\beta}}
\newcommand{\ftpG}{\overline{\gamma}}

\newcommand{\ftpAta}{\overline{\alpha}(\theta)(\alpha)}
\newcommand{\ftpBta}{\overline{\beta}(\theta)(\alpha)}
\newcommand{\ftpGta}{\overline{\gamma}(\theta)(\alpha)}

%% MODULO

\newcommand{\FOMod}{\mbox{$\mbox{\rm FO}_{\rm MOD}$}}
\newcommand{\FOOneMod}{\mbox{$\mbox{\rm FO}_{\rm MOD}^1$}}
\newcommand{\FOTwoMod}{\mbox{$\mbox{\rm FO}_{\rm MOD}^2$}}
\newcommand{\FOTwoModTrees}{\mbox{$\mbox{\rm FO}_{\rm MOD}^2[\succv, \lessv,
\succh, \lessh]$}}
\newcommand{\FOTwoModDesc}{\mbox{$\mbox{\rm FO}_{\rm MOD}^2[\succv, \lessv]$}\xspace}
\newcommand{\FOModTrees}{\mbox{$\mbox{\rm FO}_{\rm MOD}[\succv, \lessv,
\succh, \lessh]$}}
\newcommand{\FOModDesc}{\mbox{$\mbox{\rm FO}_{\rm MOD}[\succv, \lessv]$}\xspace}

\newcommand{\Mod}[1]{ {(\text{mod}\ #1)}}
\newcommand{\existsmod}[3]{\exists^{{#1} {#2}, {#3}}}
\newcommand{\ind}[2]{\mathds{1}_{#1}^{\left( #2 \right)}}

  \begin{abstract}
    We consider the satisfiability problem for the two-variable
    fragment of the first-order logic extended with modulo counting
    quantifiers and interpreted over finite words or trees. We prove a
    small-model property of this logic, which gives a technique for
    deciding the satisfiability problem. In the case of words this
    gives a new proof of \expspace upper bound, and in the case of
    trees it gives a \twoexptime algorithm.  This algorithm is
    optimal: we prove a matching lower bound by a generic reduction
    from alternating Turing machines working in exponential space; the
    reduction involves a development of a new version of tiling games.
  \end{abstract}

  \begingroup
  \let\clearpage\relax
\section{Introduction}

\subparagraph*{Two-variable logics.}  Two-variable logic, \FOt, is one
of the most prominent decidable fragments of first-order logic. It is
important in computer science because of its decidability and
connections with other formalisms like modal, temporal and description
logics or query languages. The satisfiability problem for \FOt{} is
\nexptime-complete and satisfiable formulas have models of exponential
size. In this paper we are interested in extensions of \FOt{}
interpreted over finite words and trees. It is known that \FOt{} over
words can express the same properties as unary temporal
logic~\cite{EtessamiVW02} and \FOt{} over trees is precisely as
expressive as the navigational core of XPath, a~query language for XML
documents~\cite{MarxDeRijke04}. The satisfiability problem for \FOt{}
over words is shown to be \nexptime-complete in \cite{EtessamiVW02},
and over trees---\expspace-complete in \cite{BBCKLMW-tocl}. Recently
it was shown that these complexities do not
change\cite{CharatonikWitkowski-LMCS16, BednarczykChK-CSL17} if
counting quantifiers of the form $\exists^{\leq k}, \exists^{\geq k}$
are added.

\subparagraph{Modulo counting quantifiers.} First-order logic has a
limited expressive power. For example, it cannot express even such
simple quantitative properties as parity. To overcome this problem,
Wolfgang Thomas with his coauthors introduced in the 80s of the last
century \emph{modulo counting quantifiers} of the form ``there exists
$a$ mod $b$ elements $x$ such that \ldots''. A survey of results in
first order logic extended with modulo counting quantifiers can be
found in \cite{StraubingT08}.  Recent research in this area involves a
study of definability of regular languages on words and its
connections to algebra \cite{StraubingT08}, \cite{StraubingTT95};
definable languages of $(\N, +)$ \cite{RoyS07}, \cite{KrebsS12};
equivalences of finite structures \cite{Nurmonen00}; definable tree
languages \cite{Potthoff94}, \cite{BenediktS09}; locality
\cite{HarwathS16}, \cite{HeimbergKS-LICS16}; extensions of Linear
Temporal Logic \cite{BaziramwaboMT99}, \cite{Sreejith11}, complexity
of the model-checking problem \cite{HeimbergKS-LICS16},
\cite{BerkholzKS17}. Not much is known about the complexity of the
satisfiability problem for this logic. The only work that we are aware
of is \cite{LodayaSreejith-DLT17}, which proves \expspace-completeness
of the satisfiability problem for \FOt{} with modulo counting
quantifiers over finite words. On the other hand, a simple adaptation
of automata techniques for deciding WS1S or WS2S gives a non-elementary
decision procedure for the satisfiability problem of full first-order
logic with modulo quantifiers over words or ranked trees.

\subparagraph{Our contribution.} We consider the satisfiability
problem for \FOt{} with modulo counting quantifiers interpreted over
finite words and trees. We provide an alternative to
\cite{LodayaSreejith-DLT17} proof of \expspace upper bound for the
case of words. This proof is based on a small-model property of this
logic and can be extended to the case of ordered and unranked
trees. By ordered, we mean that the list of children of any node is
ordered by the sibling relation, and by unranked, that there is no
limit on the number of children. We prove that the case of such trees
is \twoexptime-complete. For the upper bound, we again prove a
small-model theorem and then give an alternating algorithm that
decides the problem in exponential space. Since \aexpspace =
\twoexptime, this gives the desired upper bound. With some obvious
modifications this algorithm can be also applied to unordered or
ranked trees. For the lower bound, we develop a new version of tiling
games that can encode computations of alternating Turing machines
working in exponential space, and can be encoded in the logic. In our
encoding we do not use ordering of children, and the number of
children of any node is bounded by 2, which shows that already the
case of unordered and ranked trees is \twoexptime-hard.

The current paper is a~full version of~\cite{BednarczykCh-short}.

\section{Preliminaries}

%%%%%%%%%%%%%%%%%%%%%%%%%%%%%%%%%%%%%%%%%%%%%%%%%%%%%%%%%%%%%%%%%%%%%%%%%%%%%%%%

\subparagraph*{Tuples and modulo remainders.}  By a $k$-tuple of
numbers we mean an element of the set $\N^{k}$.  By $\vec{0}_k$ and
$\vec{1}_k$ we will denote $k$-tuples consisting of, respectively,
only zeros and only ones.  We denote the $i$-th element of a $k$-tuple
$t$ by $\pi_i(t)$.  We will often drop the subscript $k$ if the
dimension $k$ is clear from the context.

Consider a finite set $X$, a number $k \in \N$ and a mapping $f$ from
$X$ to $\N^k$. We say that $f$ is \emph{zero}, if for all $x \in X$ we
have $f(x) = \vec{0}_k$. We say that $f$ is a \emph{singleton}, if
there exists a unique argument $x \in X$ such that $f(x) = \vec{1}_k$
and for all other arguments the function $f$ is zero.

Let $n$ be a positive integer. By $\Z_{n}$ we denote the set of all
remainders modulo $n$, that is the set~$\set{0,1, \ldots, n{-}1}$. We
denote by $r_n$ the remainder function modulo $n$, that is $r_n : \N
\rightarrow \Z_{n}$, $r_n(x)= x\mathrm{~mod~} n$.  

%%%%%%%%%%%%%%%%%%%%%%%%%%%%%%%%%%%%%%%%%%%%%%%%%%%%%%%%%%%%%%%%%%%%%%%%%%%%%%%%

\subparagraph*{Syntax and semantics.}

We refer to structures with fraktur letters, and to their universes
with the corresponding Roman letters. We always assume that structures
have non-empty universes.  We work with signatures of the form
$\tau = \tau_0 \cup \tau_{nav}$, where $\tau_0$ is a set of unary
relational symbols, and
$\tau_{nav} \subseteq \{ \succv, \lessv, \succh, \lessh, \leq,
\wordsucc \}$ is the set of \textit{navigational} binary symbols,
which will be interpreted in a special way, depending on which kind of
structures, words or trees, are considered as models.

We define the two-variable fragment of first-order logic with modulo
counting quantifiers $\FOTwoMod$ as the set of all first-order
formulas (over the signature $\tau$), featuring only the variables $x$
and $y$, extended with modulo counting quantifiers of the form
$\existsmod{\leq}{k}{l}, \existsmod{=}{k}{l}$ and
$\existsmod{\geq}{k}{l}$, where $l \in \Z_+$ is a positive integer and
$k \in \Z_{l}$. The formal semantics of such quantifiers is as
follows: a formula $\existsmod{\bowtie}{k}{l} \; {y} \; \varphi(x,y)$,
where $\bowtie \; \in \set{\leq, =, \geq}$, is true at element $a$ of
a structure $\str{M}$, in symbols $\str{M},a \models \left(
  \existsmod{\bowtie}{k}{l} \; {y} \; \varphi(x,y)\right)$ if and only
if $ r_l \left( \; \sizeOf{\set{ b \in \mathit{M} : \varphi[a,b] }} \;
\right) \bowtie k.$ When measuring the length of formulas we will
assume binary encodings of numbers $k$ and $l$ in superscripts of
modulo quantifiers.

In \cite{LodayaSreejith-DLT17} the authors use a different notation
for modulo counting quantifiers in their logic
$\textsc{FOmod}^2[<]$. It seems that the main difference is that they
require equality in place of our $\bowtie$ comparison. The use of
$\leq$ and $\geq$ operators makes the logic exponentially more
succinct. For example the formula
$\existsmod{\geq}{2^n}{2 \cdot 2^n} x \; p(x)$ has length
$\mathcal{O}(n)$, while an equivalent formula in $\textsc{FOmod}^2[<]$
has the form
$ \bigvee_{i=2^n}^{2 \cdot 2^n} \existsmod{=}{i}{2 \cdot 2^n} x \;
p(x)$ and its length is $\Omega(2^n \cdot n)$.

We write $\FOTwoMod[\tau_{nav}]$, to denote that the only binary
symbols that are allowed in signatures are from $\tau_{nav}$.  We will
work with two logics: $\FOTwoModWords$ and $\FOTwoModTrees$.  The
former is interpreted over finite words, where $\leq$ denotes the
linear order over word positions and $\wordsucc$ is the successor
relation.  The latter logic is interpreted over finite, unranked and
ordered trees.  The interpretation of the symbols from $\tau_{nav}$ is
the following: $\succv$ is interpreted as the child relation, $\succh$
as the right sibling relation, and $\lessv$ and $\lessh$ as their
respective transitive closures.  We read $u \downarrow w$ as ``$w$ is
a \emph{child} of $u$'' and $u \rightarrow w$ as ``$w$ is the
\emph{right sibling} of $u$''.  We will also use other standard
terminology like \emph{ancestor}, \emph{descendant},
\emph{preceding-sibling}, \emph{following-sibling}, etc. In both cases
of words and trees all elements of the universe can be labeled by an
arbitrary number of unary predicates from $\tau_0$.

%%%%%%%%%%%%%%%%%%%%%%%%%%%%%%%%%%%%%%%%%%%%%%%%%%%%%%%%%%%%%%%%%%%%%%%%%%%%%%%%

\subparagraph*{Normal form.}

As usual when working with two-variable logic, it is convenient to use
so-called Scott normal form \cite{Scott1962}. The main feature of such
form is that nesting of quantifiers is restricted to depth two.  Below
we adapt this notion to the setting of $\FOTwoMod$.

\begin{definition} \label{def:fo2mod_normalform}
We say that a formula $\varphi \in \FOTwoMod$ is in \emph{normal form}, if
$$
\varphi =
\forall{x} \forall{y} \; \chi(x, y) \wedge
\bigwedge_{i=1}^{n}
\forall{x} \exists{y} \; \chi_i(x, y) \wedge
\bigwedge_{j=1}^{m}
\forall{x} \existsmod{\bowtie_j}{k_j}{l_j} {y} \; \psi_j(x, y)
$$
where $\bowtie_i \in \lbrace \leq, \geq \rbrace$,
each $l_j$ is a positive integer and each $k_j \in \Z_{l_j}$ and
all $\chi, \chi_i$ and $\psi_j$ formulas are quantifier-free.
We require both $n$ and $m$ parameters to be non-zero.
\end{definition}

The following lemma is proved by a routine adaptation of a normal form
proof from \cite{GradelOR97}.

\begin{lemma} \label{lemma:normalform}
  Let $\varphi$ be a formula from $\FOTwoMod$ over a
  signature~$\tau$. There exists a polynomially computable $\FOTwoMod$
  normal form formula $\varphi'$ over signature $\tau'$, consisting of
  $\tau$ and polynomially many additional unary symbols,
  such that $\varphi$ and $\varphi'$ are equisatisfiable.
\end{lemma}

In the next sections, when a normal form formula $\varphi$ is considered
we always assume that it is as in Definition \ref{def:fo2mod_normalform}.
In particular we allow ourselves, without explicitly recalling the shape
of $\varphi$, to refer to its  parameters $n, m$, components
$\chi, \chi_i, \psi_j$ and numbers $k_j, l_j$ given in modulo
counting quantifiers. 

Consider a conjunct $\forall{x} \exists{y} \; \chi_i(x, y)$ from
an \FOTwoMod\ normal form formula $\varphi$. Let $\str{M}$ be its 
finite model and let $v \in M$ be an arbitrary element
from the structure. Then an element $w \in M$ such that
$\str{M} \models \chi_i[v,w]$ is called a $\chi_i$-\emph{witness}
for $v$. Analogously, we will speak about $\psi_j$-witnesses for
modulo conjuncts or simply witnesses if a formula is clear from the
context.

%%%%%%%%%%%%%%%%%%%%%%%%%%%%%%%%%%%%%%%%%%%%%%%%%%%%%%%%%%%%%%%%%%%%%%%%%%%%%%%%

\subparagraph*{Order formulas.}

Let us call \emph{order formulas} the formulas specifying  relative
positions of pairs of elements in a structure with respect to the
binary navigational predicates.

In the case of words, when $\tau_{nav} = \set{\leq, \wordsucc}$, there are 
five possible order formulas: $x{=}y$, $\wordsucc(x,y)$, $\wordsucc(y,x)$,
$x{\neq}y \wedge \neg \wordsucc(x,y) \wedge x{\leq}y$ and
$x{\neq}y \wedge \neg \wordsucc(y,x) \wedge y{\leq}x$. We denote them
respectively as $\teq, \tnext, \tprev, \tl, \tg$.
Let $\Theta_{words}$ be the set of these five formulas.

In the case of trees, when $\tau_{nav} =\{ \succv, \lessv, \succh, \lessh \}$,
we use $x \fw y$ to abbreviate the formula stating that $x$ and $y$ are in 
\emph{free position}, i.e., that they are related by none of the navigational
binary predicates available in the signature. 
There are ten possible order formulas: 
$x \succv y$, $y \succv x$, $x \lessv y \wedge \neg (x \succv y)$,
$y \lessv x \wedge \neg (y \succv x)$, $x \succh y$, $y \succh x$,
$x \lessh y \wedge \neg (x \succh y)$, $y \lessh x \wedge \neg (y \succh x)$,
$x \fw y$, $x{=}y$. They are denoted, respectively, as: $\tsuccv$, $\tprecv$,
$\tlessv$, $\tgreatv$,  $\tsucch$, $\tprech$, $\tlessh$, $\tgreath$,
$\tfree$, $\teq$. Let $\Theta_{trees}$ be the set of these ten formulas. 
We will use symbol $\Theta$ to denote either $\Theta_{words}$ or
$\Theta_{trees}$ if the type of structure is not important or clear
from context.
%
%%%%%%%%%%%%%%%%%%%%%%%%%%%%%%%%%%%%%%%%%%%%%%%%%%%%%%%%%%%%%%%%%%%%%%%%%%%%%%%%
%
\subparagraph*{Types and local configurations.}

Following a standard terminology, an \emph{atomic $1$-type} over a
signature $\tau$ is a maximal satisfiable set of atoms or negated
atoms involving only the variable~$x$. Usually we identify a $1$-type
with the conjunction of all its elements.  We will denote by $\AAA$
the set of all $1$-types over $\tau$.  Note that the size of $\AAA$ is
exponential in the size of $\tau$.  For a given $\tau$-structure
$\str{M}$, and an element $v \in M$ we say that $v$ \emph{realizes} a
$1$-type $\alpha$, if $\alpha$ is the unique $1$-type such that
$\str{M} \models \alpha[v]$.  We denote by $\tp{M}{v}$ the $1$-type
realized by $v$.

A $1$-type provides us only information about a single element of a
model. Below we generalize this notion such that for each element of
a structure it provides us also some information about surrounding
elements and their $1$-types.

\begin{definition} \label{def:full_type} An
  $(l_1, l_2, \ldots, l_m)$-full type $\ftpA$ over a signature $\tau$
  is a function of the type
  $\ftpA : \Theta \rightarrow \AAA \rightarrow \set{0,1} \times
  \Z_{l_1} \times \Z_{l_2} \times \ldots \times \Z_{l_m}$ (a function
  that takes an order formula from $\Theta$ and returns a function
  that takes a $1$-type from $\AAA$ and returns a tuple of
  integers). Additionally, we require that all
  $(l_1, l_2, \ldots, l_m)$-full types satisfy the following
  conditions:
  
  \begin{itemize}
    \item $\ftpA(\teq)$ is a singleton,
    \item $\ftpA(\theta)$ is either zero or singleton for $\theta \in
      \set{ \tnext, \tprev, \tsucch, \tprecv, \tprech }$ and
    \item if $\ftpA(\tnext)$ (respectively $\ftpA(\tprev),
      \ftpA(\tsucch), \ftpA(\tsuccv), \ftpA(\tprech), \ftpA(\tprecv)$)
      is zero then also the function $\ftpA(\tg)$ (respectively
      $\ftpA(\tl), \ftpA(\tlessh), \ftpA(\tlessv), \ftpA(\tgreath),
      \ftpA(\tgreatv)$) is zero.
  \end{itemize}
\end{definition}
In the following, if the numbers $l_1, l_2, \ldots, l_m$ are clear
from the context or not important, we will be omitting them.

Let us briefly describe the idea behind the notion of full types.
Ideally, for a given element $v$ of a structure $\str{M}$, we would
like to know the exact number of occurrences of each $1$-type on each
relative position.  However, to keep the memory usage under control,
we store only a summary of this information.  For the purpose of
$\forall \exists$ conjuncts of a formula in normal form, it is enough
to know if a type $\alpha$ occurs at least once in a given relative
position.  This information is stored in the 0-1 part of a full type.
Additionally, for the purpose of $\existsmod{\bowtie_j}{k_j}{l_j}$
conjuncts, we store the remainders of the total numbers of occurrences
of each $1$-type $\alpha$ modulo each of $l_j$ appearing in modulo
quantifiers.

\begin{definition}\label{def:realized_ftp}
  Let $\str{M}$ be a $\tau$-structure and $v$ an arbitrary element
  from $M$.  We will denote by
  $(l_1, l_2, \ldots, l_m)$--$\ftp{M}{}{v}$ the unique
  $(l_1, l_2, \ldots, l_m)$-full type \emph{realized} by $v$ in
  $\str{M}$, i.e., the $(l_1, l_2, \ldots, l_m)$-full type $\ftpA$
  such that for all order formulas $\theta \in \Theta$ and for all
  atomic $1$-types $\alpha$, the value of $\ftpA(\theta)(\alpha)$ is
  equal to the tuple
  $\left( \mathit{cut}_1(W), r_{l_1}(W), r_{l_2}(W), \ldots,
    r_{l_m}(W) \right)$, where $W$ is the number of occurrences of
  elements of type $\alpha$ in the relative position described by
  $\theta$, namely
$$
W=\sizeOf{\set{ w \in M : \str{M} \models \theta[v,w] \wedge
    \tp{M}{w}=\alpha}}.
$$
and $\mathit{cut}_1(W)$ is $0$ if $W$ is empty and $1$ otherwise.
\end{definition}

%%%%%%%%%%%%%%%%%%%%%%%%%%%%%%%%%%%%%%%%%%%%%%%%%%%%%%%%%%%%%%%%%%%%%%%%%%%%%%%%

The following definition and lemma are basic tools used in the proofs
of small-model properties $\FOTwoMod$.

\begin{definition}[$\varphi$-consistency] \label{def:phiconsistent}
  Let $\varphi$ be a $\FOTwoMod$ formula in normal form and let $\alpha$ be
  the unique $1$-type satisfying $\ftpA(\teq)(\alpha)=\vec{1}$. We say that
  a $(l_1, l_2, \ldots, l_m)$-full type $\ftpA$ is
  \emph{$\varphi$-consistent}, if it satisfies the following
  conditions:
\begin{enumerate}
\item It does not violate the $\forall \forall$ subformula i.e. for
  all order formulas $\theta \in \Theta$ and for all $1$-types
  $\beta$ such that 
  $\ftpA(\theta)(\beta) {\neq} \vec{0}$, the implication $\alpha(x)
  {\wedge} \beta(y) {\wedge} \theta(x,y) \models \chi(x,y)$ holds.
\item There is a witness for each $\forall \exists$ conjunct of
  $\varphi$.  Formally, consider a number $1 \leq i \leq n$ and a
  subformula $\forall{x} \exists{y} \; \chi_i(x,y)$.  We
  require that there exists a $1$-type $\beta$ and an order formula $\theta
  \in \Theta$ such that $\ftpA(\theta)(\beta) {\neq} \vec{0}$ and the
  logical implication $\alpha(x) {\wedge} \beta(y) {\wedge}
  \theta(x,y) \models \chi_i(x,y)$ holds.
\item Each modulo conjunct has the right number of witnesses. Consider
  a number $1 \leq j \leq m$ and a subformula $\forall{x}
  \existsmod{\bowtie_j}{k_j}{l_j} {y} \; \psi_j(x, y)$. We
  require that the remainder modulo $l_j$ of the number of witnesses
  encoded in the full-type satisfies the inequality $\bowtie_j
  k_j$. Formally,
$$
r_{l_j}
\Bigg(
\sum_{
  \theta \in \Theta, \beta \in \AAA \; : \;
  \alpha(x) {\wedge} \beta(y) {\wedge} \theta(x,y) \models \psi_j(x,y)
  } \pi_{j+1} \left( \ftpA(\theta)(\beta) \right)
\Bigg)
\bowtie_j k_j. 
$$
\end{enumerate}
\end{definition}

A proof of the following lemma is a straightforward unfolding of Definitions
\ref{def:realized_ftp} and \ref{def:phiconsistent} and the definition
of the semantics of \FOTwoMod.
 
\begin{lemma} \label{lemma:ficonsistent} Assume that a formula
  $\varphi \in \FOTwoMod$ in normal form is interpreted over finite
  words or trees. Then $\str{M} \models \varphi$ if and only if every
  $(l_1, l_2, \ldots, l_m)$-full type realized in $\str{M}$ is
  $\varphi$-consistent.
\end{lemma}

Let $\varphi$ be $\FOTwoMod$ formula in normal form. For proving the upper bounds in
the next sections, we will estimate the size of a model by the following function. 
We define $\str{f}(\varphi)$ as the function, which for a given formula returns 
the total number of $(l_1, l_2, \ldots, l_m)$-full types over the
signature of $\varphi$.
To be precise, $\str{f}(\varphi)$ is equal to
$(2 l_1 l_2 \ldots l_m)^{|\Theta| |\AAA|}$. Note that $\str{f}(\varphi)$ is doubly
exponential in the size of the formula $\varphi$.

%%%%%%%%%%%%%%%%%%%%%%%%%%%%%%%%%%%%%%%%%%%%%%%%%%%%%%%%%%%%%%%%%%%%%%%%%%%%%%%%

\section{Finite words}
\label{sec:word_algo}
In this section we focus our attention on the case of finite words.
We give an alternative proof of \expspace upper bound for the
satisfiability problem of \FOTwoModWords. Originally this result was
proved in \cite{LodayaSreejith-DLT17} by a reduction to unary temporal
logic with modulo counting operators. Here we give a direct algorithm
dedicated to \FOTwoModWords, based on a small model property that we
prove. The advantage of the method is that it allows an extension to
the case of trees and other extensions (e.g., an incorporation of
counting quantifiers of the form $\exists^{\leq k}, \exists^{\geq k}$
is quite obvious).

%%%%%%%%%%%%%%%%%%%%%%%%%%%%%%%%%%%%%%%%%%%%%%%%%%%%%%%%%%%%%%%%%%%%%%%%%%%%%%%%

\subparagraph*{Small model property.} 
We start by proving that satisfiable formulas in \FOTwoModWords\
have  small models. The proof technique is similar to the pumping lemma
known from the theory of finite word automata. 

\begin{lemma} \label{lemma:words_small_model}
Every normal form $\FOTwoModWords$ formula satisfiable over finite words
has a model $\str{W}$ of size bounded by $\str{f}(\varphi)$.
\end{lemma}

\begin{proof}
  Consider a satisfiable formula $\varphi \in \FOTwoModWords$ and
  assume that its model $\str{W}$ is a word longer than
  $\str{f}(\varphi)$.  We will show that we can remove some subword
  from $\str{W}$ and obtain a shorter model.  By repeating this
  process, we will finally obtain a model of $\varphi$ with a required
  size.

  By the pigeonhole principle there exist two positions $u, v \in W$
  with equal full-types.  Let $\str{W}'$ be a word obtained from
  $\mathfrak{W}$ by removing all letters from positions between $u$
  and $v$ and collapsing $u$ and $v$ into a single position. Observe
  that since full types of $u$ and $v$ are equal, for all $j$ the
  remainders modulo $l_j$ of the total number of removed $1$-types on
  positions between $u$ and $v$ are equal to $0$.  Note also that due
  to presence of $0$-$1$ part in the definition of a full type, all unique
  $1$-types realized by the structure survive the surgery.  Therefore,
  all full types in $\str{W}$ remain unchanged. Since all these full
  types were $\varphi$-consistent in $\str{W}$, they are also
  $\varphi$-consistent in $\str{W}'$.  As a consequence of
  Lemma~\ref{lemma:ficonsistent}, the word $\str{W}'$ is indeed a
  model of $\varphi$, as expected.
\end{proof}

%%%%%%%%%%%%%%%%%%%%%%%%%%%%%%%%%%%%%%%%%%%%%%%%%%%%%%%%%%%%%%%%%%%%%%%%%%%%%%%%

\subparagraph*{Algorithm.} 

Now we are ready to present an $\expspace$ algorithm for solving the
satisfiability for $\FOTwoModWords$ interpreted over finite words.  We
assume that the input formula $\varphi$ is in normal form. Since
models of $\varphi$ can have doubly-exponential length, we cannot
simply guess a complete intended model.  To overcome this difficulty,
we guess the model on the fly, letter by letter. For each position of
the guessed word we guess its full type and after checking some
consistency conditions described in Definition
\ref{def:phiconsistent}, we verify if it can be linked with the
previous position. This way we never have to store in memory more than
two full types. Since the size of a full type is bounded exponentially
in $\sizeOf{\varphi}$, the whole procedure runs in \expspace. We accept the
input, if the guessing process ends after at most $\str{f}(\varphi)$
steps.

To avoid presentational clutter in the description of our algorithm we
omit the details of two simple tests \textit{is-valid-successor} and
\textit{is-$\varphi$-consistent}. The former takes two full types,
$\ftpB$ and $\ftpA$, and checks if the position of the type $\ftpA$
can indeed have a successor of type $\ftpB$. It can be easily done
by comparing the total number of $1$-types on each relative position
stored in both types.  The latter test takes one full type $\ftpA$ and
checks if it satisfies the conditions described in Definition
\ref{def:phiconsistent}.

\begin{algorithm}[H] \label{algo:sat_test_fo2mod_words}
  \DontPrintSemicolon
  \KwData{Formula $\varphi \in \FOTwoModWords$ in normal form.} % Input
  \caption{Satisfiability test for $\FOTwoModWords$}
  
  $\text{MaxLength} := \str{f}(\varphi)$
  \tcp*{Maximal length of a model from Lemma \ref{lemma:words_small_model}}
  
  $\text{CurrentPosition} := 0$
  
  \textbf{guess} a full type $\ftpA$
  s.t. $\ftpA(\tl)$ and $\ftpA(\tprev)$ are zero
  \tcp*{Type of the first position}
  
  \textbf{if not} is-$\varphi$-consistent$(\ftpA)$
  \textbf{then reject}
  \tcp*{See Definition \ref{def:phiconsistent}}
  
  \While{$\text{CurrentPosition} < \text{MaxLength}$}{

  \textbf{if} both $\ftpA(\tg)$ and $\ftpA(\tnext)$
  are zero \textbf{then} \textbf{accept}
  \tcp*{Type of the last position}
  
  \textbf{guess} a full type $\ftpB$
  \tcp*{Type of the successor}
  
  \textbf{if not} is-$\varphi$-consistent$(\ftpB)$
  \textbf{then reject}
  \tcp*{See Definition \ref{def:phiconsistent}}
  
  \textbf{if not} is-valid-successor$(\ftpB, \ftpA)$
  \textbf{then reject}
  
  $\ftpA := \ftpB$
  
  $\text{CurrentPosition} := \text{CurrentPosition} + 1$
  }
  
  \textbf{reject}
\end{algorithm}
The correctness of the algorithm above is guaranteed by the following
lemma.%, which is proved in Appendix \ref{app:wordsCorrect}.

\begin{lemma}\label{lem:wordsCorrect}
  Procedure \ref{algo:sat_test_fo2mod_words} accepts its input
  $\varphi$ if and only if $\varphi$ is satisfiable.
\end{lemma}

\begin{proof}
  Take an arbitrary satisfiable formula $\varphi \in \FOTwoModWords$.
  Lemma \ref{lemma:words_small_model} guarantees the existence of a
  small model $\str{W}$ of length not exceeding $\str{f}(\varphi)$.
  At each step of the algorithm we can guess exactly the same types as
  in $\str{W}$.  Such types are of course $\varphi$-consistent and
  valid in the sense of \textit{is-valid-successor}, so the procedure
  accepts $\varphi$.
  
  For the opposite direction, assume that Procedure
  \ref{algo:sat_test_fo2mod_words} accepts its input
  $\varphi \in \FOTwoModWords$. Then we can reconstruct the word
  $\str{W}$ from the guessed full-types. Due to fact that all guessed
  full-types are $\varphi$-consistent we have the right number of
  witnesses to satisfy the formula. Additionally,
  \textit{is-valid-successor} function guarantees that two consecutive
  positions are linked correctly. So the guessed structure is indeed a
  word in which every realized type is $\varphi$-consistent. By Lemma
  \ref{lemma:ficonsistent}, the word $\mathfrak{W}$ is a model of
  $\varphi$.
\end{proof}

\section{Finite trees}

In this section we prove the main result of this paper, which is the
following theorem. It is a direct consequence of
Theorems~\ref{th:upper} and \ref{th:lower} below.

\begin{theorem}
  The satisfiability problem for $\FOTwoModTrees$ interpreted over
  finite trees is \twoexptime-complete.
\end{theorem}

%%%%%%%%%%%%%%%%%%%%%%%%%%%%%%%%%%%%%%%%%%%%%%%%%%%%%%%%%%%%%%%%%%%%%%%%%%%%%%%%

\subsection{Upper bound}

We start by proving the \twoexptime upper bound. As in previous
section, this is done by showing first a small model property and then
an algorithm. The small model property is crucial in the proof of
correctness of the algorithm. Actually, having defined the notion of
the full type, the technique here is a rather straightforward
combination of the technique from previous section and from
\cite{BBCKLMW-tocl}. The main difficulty here was to come up with the
right notion of a full type.

\subsubsection{Small model property}

We demonstrate the small-model property of the logic $\FOTwoModTrees$
by showing that every satisfiable formula $\varphi$ has a tree model
of depth and degree bounded by $\str{f}(\varphi)$. This is done by
first shortening $\downarrow$-paths and then shortening the
$\rightarrow$-paths, as in the proof of
Lemma~\ref{lemma:words_small_model}. 

\begin{theorem}[Small model theorem]
  \label{thm:fo2modtrees_smallmodeltheorem}
  Let $\varphi$ be a normal form $\FOTwoModTrees$ formula.  If
  $\varphi$ is satisfiable then it has a a tree model in which every
  path has length bounded by $\str{f}(\varphi)$ and every vertex has
  degree bounded by $\str{f}(\varphi)$.
\end{theorem}

\begin{proof} Let $\str{T}$ be a model of $\varphi$.  First we show
  how to shorten $\downarrow$-paths in $\str{T}$. Assume that there
  exists a $\downarrow$-path in $\str{T}$ longer than
  $\str{f}(\varphi)$. Thus, by the pigeonhole principle, there are two
  nodes $u$ and $v$ on this path such that $v$ is a descendant of $u$
  and $\ftp{T}{\varphi}{u} = \ftp{T}{\varphi}{v}$. Consider the tree
  $\str{T}'$, obtained by removing the subtree rooted at $u$ and
  replacing it by the subtree rooted at $v$.
  
  Observe that for all $j$ the remainders modulo $l_j$ of the total
  number of removed vertices is equal to $0$. Additionally, all unique
  full types survive the surgery. Thus, all full types are retained
  from the original tree, so they are $\varphi$-consistent and by
  Lemma \ref{lemma:ficonsistent} the tree $\str{T}'$ is a model of
  $\varphi$. By repeating this  process we  get a tree with
  all $\downarrow$-paths shorter than $\str{f}(\varphi)$.

  Shortening the $\rightarrow$-paths is done in a similar way.  Assume
  that there exists a node $v$ with branching degree greater than
  $\str{f}(\varphi)$. Then there exist two children $u, w$ of $v$,
  with equal full types. Let $\str{T}''$ be a tree obtained by
  removing all vertices between $u$ and $w$ (excluding $u$ and
  including $w$) together with subtrees rooted at them. Again, the
  remainders modulo $l_j$ of the total number of removed vertices are
  all equal to $0$, and all unique types survive the surgery, so all
  full types are retained from $\str{T}$ and thus
  $\varphi$-consistent.  Therefore the tree $\str{T}''$ is a model of
  $\varphi$.  We repeat this process until we get a tree with desired
  branching.
\end{proof}

%%%%%%%%%%%%%%%%%%%%%%%%%%%%%%%%%%%%%%%%%%%%%%%%%%%%%%%%%%%%%%%%%%%%%%%%%%%%%%%%

\subsubsection{Algorithm}

In this section, we design an algorithm to check if a given
$\FOTwoModTrees$ formula $\varphi$ interpreted over finite trees is
satisfiable.  By Lemma \ref{lemma:normalform} we can assume that the
input formula $\varphi$ is given in normal form. Then, by Theorem
\ref{thm:fo2modtrees_smallmodeltheorem}, we can turn our attention to
trees with degree and path length bounded by
$\str{f}(\varphi)$. Recall that $\str{f}(\varphi)$ is
doubly-exponential in $\sizeOf{\varphi}$.

Let us describe the core ideas of the algorithm. It works in
alternating exponential space. Since \aexpspace = \twoexptime, it can
be translated to an algorithm working in \twoexptime.  The algorithm
starts by guessing the full type of the root and then it calls a
procedure that builds a tree with the given full type of the root.

The procedure called on a vertex $v$ guesses on the fly the children
of $v$, starting with the leftmost child and storing in the memory at
most two children at the same time. While guessing a child $v'$, it
checks that its type is $\varphi$-consistent and that it can be
correctly linked to $v$ and to its left sibling (if it exists). Then
the procedure is called recursively on $v'$ to build the subtree
rooted at $v'$, but the recursive call is done in parallel, exploiting
the power of alternation.  The process is continued in the same way,
until the bottom of the tree is reached or the height of the
constructed tree exceeds $\str{f}(\varphi)$.

There is one more point of the algorithm, namely, we have to make sure
that the types of children and their descendants, guessed on the fly,
are consistent with the information stored in their parent's full
type, in particular with $\tsuccv$ and $\tlessv$ components. To handle
the first of them, we can simply compute the union of all $\teq$
components of children, which can be done during the guessing
process. Analogously, to handle the second case, we simply compute the
union of all $\tsuccv$ and $\tlessv$ components of all children's full types.

Below we present a pseudocode of the described procedure.  Similarly
to the case of finite words, we omit the details of obvious methods
for checking consistency. For calculating the union of full-types
mentioned in the previous paragraph, we employ $\oplus$ operation
defined as follows. For given $(l_1, l_2, \ldots, l_m)$-full types
$\ftpA$ and $\ftpB$, the result of $\ftpA \oplus \ftpB$ is a
$(l_1, l_2, \ldots, l_m)$-full type $\ftpG$ such that for all order
formulas $\theta \in \Theta$ and all $1$-types $\alpha \in \AAA$, the
following condition holds:
$$
\ftpGta = \left( \max \left( \pi_0(\ftpAta), \pi_0(\ftpBta) \right), R_1, R_2, \ldots, R_m \right),
$$
where $R_i = r_{l_i}(\pi_i(\ftpAta) + \pi_i(\ftpBta))$.

\begin{algorithm}[htb]
  \DontPrintSemicolon
  \caption{Building a subtree rooted at given node}
   \label{algo:subtree_build}
  \KwData{
    Formula $\varphi \in \FOTwoModTrees$ in normal form,\\
    full type $\ftpA$ of a starting node,
    and current level $\text{Lvl} \in \mathbb{N}$.
  } % Input 
  
  \textbf{if not} is-$\varphi$-consistent$(\ftpA)$
  \textbf{then reject}
  \tcp*{See Definition \ref{def:phiconsistent}}
  
  \textbf{if} $\text{Lvl} \geq \str{f}(\varphi)$
  \textbf{then reject}
  \tcp*{Path too long}
  
  \textbf{if} \label{linia3} $\ftpA(\tsuccv)$ is zero
  \textbf{then accept}
  \tcp*{Last node on the path}  
   
  \textbf{Guess} \label{linia4} the degree $\text{Deg} \in [1, \str{f}(\varphi)]$ of a node\\
  
  \textbf{Guess} the full type $\ftpB$ of the leftmost child
  and check if its a valid leftmost son of $\ftpA$

  $O_{\tsuccv} := \ftpB(\teq)$
  \tcp*{Types  of children guessed so far}
  
  $O_{\tlessv} := \ftpB(\tsuccv) \oplus \ftpB(\tlessv)$
  \tcp*{Types  of  descendants guessed so far}

  \While{$\text{Deg} > 1$}{
      Run in parallel $\text{Procedure}$ \ref{algo:subtree_build} on
      $(\varphi, \ftpB, \text{Lvl}+1)$
      \tcp*{Alternation here}

    \textbf{Guess} a full type $\ftpG$ of the right brother
      of $\ftpB$ and check consistency with $\ftpA$
            
      $O_{\tsuccv} := O_{\tsuccv} \oplus \ftpG(\teq)$, 
      $O_{\tlessv} := O_{\tlessv} \oplus
        \ftpG(\tsuccv) \oplus \ftpG(\tlessv)$ 
      \tcp*{Updating obligations}

      $\ftpB := \ftpG$,
      $\text{Deg} := \text{Deg} - 1$
  }
  
  Run in parallel $\text{Procedure}$ \ref{algo:subtree_build} on
      $(\varphi, \ftpB, \text{Lvl}+1)$
      \tcp*{Last child}  

  \textbf{if} $\ftpB(\tsucch)$ is not zero
  \textbf{then reject}
  \tcp*{Not valid last node on $\rightarrow$-path.}
  
  \textbf{if} \label{linia15} $\ftpA(\tsuccv) = O_{\tsuccv}$ and
    $\ftpA(\tlessv) = O_{\tlessv}$
  \textbf{then accept else reject}

\end{algorithm}

\begin{algorithm} % google mowi: always place the \label after the \caption
  \DontPrintSemicolon
  \KwData{Formula $\varphi \in \FOTwoModTrees$ in normal form.} % Input
  \caption{Satisfiability test for $\FOTwoModTrees$}
  \label{alg:sat_trees}
  \textbf{guess} \label{proc:linia1} a full type $\ftpA$
  s.t. $\ftpA(\tprecv)$ is zero
  \tcp*{Type of the root}
  
  Run $\text{Procedure}$ \ref{algo:subtree_build} on
      $(\varphi, \ftpA, 1)$

\end{algorithm}

The following lemma guarantees the correctness of the algorithm and
leads directly to the main result of this section.
\begin{lemma} \label{lem:treesCorrect}
Procedure \ref{alg:sat_trees} accepts its input formula
$\varphi \in \FOTwoModTrees$ if and only if $\varphi$ is satisfiable
over  finite trees.  
\end{lemma}

\begin{proof}
Assume that a given normal form formula $\varphi \in \FOTwoModTrees$
has a tree model $\str{T}$. Then by Theorem
\ref{thm:fo2modtrees_smallmodeltheorem}
this tree can be transformed into a tree model $\str{T}'$ with both
branching and  path length bounded by $\str{f}(\varphi)$. At each step
of the algorithm we can guess exactly the same full-types as in $\str{T}'$. 
Guessed types are $\varphi$-consistent and do not violate any of the 
propagation conditions given in the algorithm.
Hence $\varphi$ is accepted by the Procedure~\ref{alg:sat_trees}.

For the opposite direction assume that Procedure \ref{alg:sat_trees}
accepts the input formula $\varphi$. Below we construct a tree
structure $\str{T}$ and show that it is indeed a model of $\varphi$.
We start by constructing the the root and labeling it with unary
symbols as required by the 1-type $\ftpA(\teq)$ guessed in line
\ref{proc:linia1} of Procedure~\ref{alg:sat_trees}.  Then we construct
as many children as guessed by Procedure~\ref{algo:subtree_build} in
line~\ref{linia4}, and recursively construct the subtrees rooted at
them. At each step of the construction we take care that the types
realised by the constructed nodes coincide with the types guessed by
the procedure. The input type $\ftpA$ in Procedure~\ref{alg:sat_trees}
forms a~kind of obligation: we must construct a node that realises
this type. Most components of this obligation are fulfilled
immediately by the construction of the children, the $\tlessv$
component is an exception. Partly it is fulfilled directly by
$\tsuccv$ components of the children, partly it is passed in $\tlessv$
components of the guessed types of the children as further obligations
to be fulfilled in future. Line~\ref{linia15} guarantees local
consistency between obligations of the current node and its
children. Line \ref{linia3} guarantees that there are no more
obligations at the end of each path.

The coincidence between guessed and realised types guarantees that the
realised types are $\varphi$-consistent, and thus by Lemma
\ref{lemma:ficonsistent} the constructed tree is indeed a model of
$\varphi$.
\end{proof}

\begin{theorem}
  \label{th:upper}
  The satisfiability problem for $\FOTwoModTrees$ interpreted over
  finite trees is in \twoexptime.
\end{theorem}

\subsection{Lower bound}
In this section we prove that the satisfiability of \FOTwoModDesc is
\twoexptime-hard. We exploit here the fact that \twoexptime=\aexpspace
and provide a generic reduction from \aexpspace. This is done in two
steps. First we translate computations of alternating Turing machines
to winning strategies in (our version of) tiling games. These
strategies are then encoded as trees and their existence is translated
to the satisfiability problem in \FOTwoModDesc.
\subsubsection{Alternating Turing machines.}
\label{sec:atm}
Let us first fix the notation concerning alternating Turing
machines. By an alternating Turing machine we understand a tuple of
the form $\tuple{Q,E,q_0,q_A,q_R,\Sigma,\Gamma,\delta}$, where $Q$ is
a finite set of states; $E\subseteq Q$ is a set of \emph{existential}
states; $q_0\in E$ is the \emph{initial} state; $q_A,q_R\in Q$ are
respectively \emph{accepting} and \emph{rejecting} states; $\Gamma$ is
a finite \emph{alphapet} of \emph{tape} symbols;
$\Sigma\subseteq \Gamma$ is an \emph{input} alphabet;
$\delta\subseteq Q\times\Gamma\times\Gamma\times\{L,N,R\}\times Q$ is
the \emph{transition} relation. States in $Q\setminus E$ are called
\emph{universal}. While in general alternating machines may have
multiple accepting and rejecting states, we assume here without loss of
generality, that there is exactly one accepting and exactly one
rejecting state. In the following we will use, also without loss of
generality, other simplifications. We assume that for each non-final
(that is, non-accepting and non-rejecting) state $q$ and for each tape
symbol $a\in \Gamma$ there are exactly two transitions of the form
$\tuple{q,a,\_,\_,\_}$ in $\delta$ and that these transitions are
alternating between universal and existential states, so that 
$\delta\subseteq \big(E\times\Gamma\times\Gamma\times\{L,N,R\}\times
(Q{\setminus} E)\big)\cup \big((Q{\setminus}
E)\times\Gamma\times\Gamma\times\{L,N,R\}\times E\big)$.
The machine has one tape that initially contains an input word. We are
interested in machines working in exponential space; again without loss
of generality we assume that given an input word of length $n$ such
a~machine never uses more then $2^n-1$ tape cells, so its
configuration is described by a~word of length exactly $2^n$ in the
language $\Gamma^*Q\Gamma^+$. A~configuration of the form $w_1qw_2$
means that the tape of the machine contains the word $w_1w_2$, the
machine is in state $q$ with its head scanning the first letter of
$w_2$. The initial configuration with an input word $w$ is
$q_0wB^{2^n{-}n{-}1}$ where $B\in\Gamma$ is the blank symbol. The machine moves
from a configuration $w_1qaw$ to its successor in accordance with the
transition relation: it chooses a transition of the form
$\tuple{q,a,b,m,q'}$ from $\delta$ (recall that there should be
exactly two such transitions if $q$ is not final), replaces the symbol
$a$ with $b$ on the tape, changes its state to $q'$ and moves the head
to the left (if $m=L$), to the right (if $m=R$) or leaves the head in the
same position (if $m=N$). The input word is \emph{accepted} if the
initial configuration \emph{leads to acceptance}, where leading to
acceptance is defined recursively as follows: accepting configurations
(i.e., configurations containing the accepting state) lead to
acceptance; an existential configuration leads to acceptance if there
exists its successor configuration that leads to acceptance; a
universal configuration leads to acceptance if all its successor
configurations lead to acceptance.

\subsubsection{Tilling Games}
Corridor tiling games, aka rectangle tiling games~\cite{Chlebus86}
provide a well-known technique for proving lower bounds in space
complexity. Here we develop our own version of these games that is
able to encode alternating Turing machines from previous section and
to be encoded in \FOTwoModDesc.

By a tiling game we understand a tuple of the form
$\tuple{C,T_0,T_1,n,\tuple{t_0,\ldots, t_{n}},\Box, L}$, where $C$ is
a finite set of \emph{colors}; $T_0,T_1\subseteq C^4$ are two sets of
\emph{tiles} (these two sets are not meant to be disjoint); $n$ is a
natural number; $\tuple{t_0,\ldots, t_{n}}$ is an \emph{initial}
sequence of $n+1$ tiles; $\Box\in C$ is a special color called
\emph{white}; $L\subseteq T_0{\cup} T_1$ is a set of tiles allowed in
the \emph{last} row.

We think of a tile $\tuple{a,b,c,d}\in C^4$ as of a square consisting
of four smaller squares, colored respectively with colors $a,b,c$ and
$d$ (see Figure~\ref{fig:tiles}). In the following we will require
that adjacent tiles have matching colors, both horizontally and
vertically. Formally, we define the horizontal adjacency relation
$H=\{\tuple{\tuple{a,b,c,d},\tuple{b,e,d,f}}\mid a,b,c,d,e,f,\in C\}$
and the vertical adjacency relation
$V=\{\tuple{\tuple{a,b,c,d},\tuple{c,d,e,f}}\mid a,b,c,d,e,f,\in C\}$.
We define a \emph{correctly tiled corridor} to be a rectangle of size
$k \times 2^n$ for some $k\in\N$ filled with tiles, with all
horizontally adjacent tiles in $H$ and all vertically adjacent tiles
in $V$, with first row starting with tiles $t_0,\ldots, t_{n}$ and
padded out with white tiles, with last row built only from tiles in
$L$, and with all edges being white.  Figure~\ref{fig:tiles} shows an
example of correctly tiled corridor for $n=2$ and initial tiles
$\tuple{\Box,\Box,\Box,a}, \tuple{\Box,\Box,a,b},
\tuple{\Box,\Box,b,\Box}$.

\newcommand{\sqtile}[4]{
\begin{tabular}[c]{|@{\hspace{1mm}}c@{\hspace{1mm}}|@{\hspace{1mm}}c@{\hspace{1mm}}|}
  \hline \makebox[1ex]{#1}&\makebox[1ex]{#2}\\ \hline
  \makebox[1ex]{#3}&\makebox[1ex]{#4}\\\hline    
\end{tabular}}

\begin{figure}[htb]
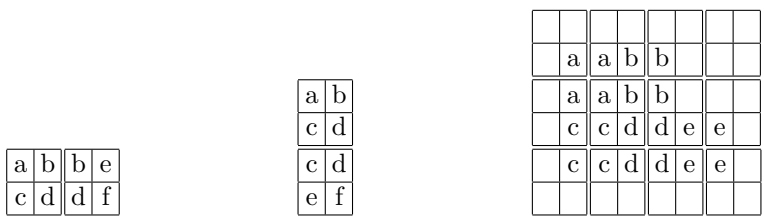

\hfill
  \begin{minipage}[b]{1.5cm}
    \sqtile{a}{b}{c}{d}\hspace{-2pt}\sqtile{b}{e}{d}{f}
  \end{minipage}\hfill
  \begin{minipage}[b]{.75cm}
        \sqtile{a}{b}{c}{d}\\ \sqtile{c}{d}{e}{f}
      \end{minipage}\hfill
      \begin{minipage}[b]{3cm}
            \sqtile{~}{~}{~}{a}\hspace{-2pt}\sqtile{~}{~}{a}{b}\hspace{-2pt}%
            \sqtile{~}{~}{b}{~}\hspace{-2pt}\sqtile{~}{~}{~}{~}\hspace{-2pt}\\
            \sqtile{~}{a}{~}{c}\hspace{-2pt}\sqtile{a}{b}{c}{d}\hspace{-2pt}%
            \sqtile{b}{~}{d}{e}\hspace{-2pt}\sqtile{~}{~}{e}{~}\hspace{-2pt}\\
            \sqtile{~}{c}{~}{~}\hspace{-2pt}\sqtile{c}{d}{~}{~}\hspace{-2pt}%
            \sqtile{d}{e}{~}{~}\hspace{-2pt}\sqtile{e}{~}{~}{~}\hspace{-2pt}
      \end{minipage}\hfill~
  \caption{Horizontally  adjacent tiles,  vertically adjacent tiles, and a correct tiling}
  \label{fig:tiles}
\end{figure}

Consider the following game. There are two players, \emph{Prover}
(also called the \emph{existential} player) and \emph{Spoiler} (also
called the \emph{universal} player). The task of Prover is to
construct a correct tiling; the task of Spoiler is to prevent Prover
from doing this. At the beginning they are given the initial row
$t_0,\ldots,t_n, (\Box^4)^{2^n -n-1}$. In each move the players
alternately choose one of the two sets $T_0$ or $T_1$ and build one
row consisting of $2^n$ tiles from the chosen set. The first move is
performed by Prover. Prover wins if after a finite number of moves
there is a correctly tiled corridor, otherwise Spoiler wins. There are
two possibilities for Spoiler to win: either Prover cannot make a move
while the last constructed row is not in $L^*$ or the game lasts
forever.

We say that a tiling game $\tuple{\_,T_0,T_1,\_,\_ ,\_, \_}$ is
\emph{well-formed} if for any play of the game and for any partial
(i.e., non-complete) tiling constructed during this play, exactly one
new row can be correctly constructed from tiles in $T_0$ and exactly
one from tiles in~$T_1$. In other words, every possible move in any
play of the game is fully determined by the choice of the set of
tiles, $T_0$ or $T_1$.

\subsubsection{From Alternating Machines to Tiling Games}

The first step of the lower-bound proof for \FOTwoModDesc is a
reduction from alternating machines to tiling games.  Actually, we
have defined our version of tiling games in such a way that the proof
of the following theorem becomes a routine exercise.

\begin{theorem} \label{th:TMtoGame}
  For all alternating Turing machines $M$ working in exponential space
  and all input words $w$ there exists a well-formed tiling game of
  size polynomial in the sizes of $M$ and $w$ such that Prover has a
  winning strategy in the game if and only if $w$ is accepted by $M$.
\end{theorem}

\begin{proof}[Proof of Theorem \ref{th:TMtoGame}] Assume that
  $M=\tuple{Q,E,q_0,q_A,q_R,\Sigma,\Gamma,\delta}$ satisfies all
  restrictions made in Section~\ref{sec:atm}. Let $n$ be the length of
  $w$ and let $w=w_0\ldots w_{n-1}$. We construct a game
  $\tuple{C,T_0,T_1,n,\tuple{t_0,\ldots, t_{n}},\Box, L}$, where
  $C=\Gamma\cup (Q\times\Gamma)\cup \{\Box\}$. The initial sequence is
  defined as follows: $t_0=\tuple{\Box,\Box,\Box,\tuple{q_0,w_0}}$,
  $t_1=\tuple{\Box,\Box,\tuple{q_0,w_0},w_1}$,
  $t_i=\tuple{\Box,\Box,w_{i-1},w_i}$ for $i\in\{2,\ldots,n-1\}$ and
  $t_n=\tuple{\Box,\Box,w_n,\Box}$. It simply defines the upper white
  edge of the constructed corridor with the initial configuration of
  $M$ stored in lower-right squares of tiles $t_0,\ldots,t_{n-1}$,
  where both the current state $q_0$ of $M$ and the scanned symbol
  $w_0$ are encoded as one color $\tuple{q_0,w_0}$. The set $L$
  contains all tiles used in representations of accepting
  configurations of $M$, with the lower edge white. Formally,
\[\begin{array}{rcccccc}
  L&=&
  (\Gamma\cup(\{ q_A\}{\times} \Gamma)\cup\{\Box\})&\times&(\Gamma\cup\{\Box\})&\times\{\Box\}\times\{\Box\}\\
 &\cup & (\Gamma\cup\{\Box\})&\times&(\Gamma\cup(\{ q_A\}{\times}\Gamma)\cup\{\Box\})&\times\{\Box\}\times\{\Box\}.
\end{array}
\]

Now we define the two sets $T_0$ and $T_1$. These are two sets of
tiles used for encoding of transitions of $M$; we use here two sets to
separate the encodings of the two transitions applicable in each
non-final configuration.  Both sets $T_0$ and $T_1$ contain tiles of
the form $\tuple{a,b,a,b}$ for all $a,b\in\Gamma\cup\{\Box\}$. This
allows copying the content of a tape cell of $M$ from one
configuration to its successor in positions distant from the position
of the head of $M$. Both sets $T_0$ and $T_1$ contain tiles
corresponding to final configurations: all tiles from L and a similar
set of tiles for rejecting configurations, with $q_A$ replaced by
$q_R$. To finish the construction of sets $T_0$ and $T_1$, for all
non-final states $q$ and all tape symbols $a\in\Gamma$, the set $T_0$
(respectively, $T_1$) contains the tiles corresponding to the first
(respectively, the second) transition of the form
$\tuple{q,a,\_,\_,\_}$ in $\delta$, defined as follows.

For all $x,y,z\in \Gamma\cup\{\Box\}$ there are three tiles
corresponding to a transition of the form $\tuple{q,a,b,L,q'}$, namely
$\tuple{x,y,x,\tuple{q',y}}$, $\tuple{y,\tuple{q,a},\tuple{q',y},b}$
and $\tuple{\tuple{q,a},z,b,z}$ (see Figure~\ref{fig:trans}).  Tiles
corresponding to a transition of the form $\tuple{q,a,b,R,q'}$ are
$\tuple{x,\tuple{q,a},x,b}$, $\tuple{\tuple{q,a},y,b,\tuple{q',y}}$
and $\tuple{y,z,\tuple{q',y},z}$ for all
$x,y,z\in \Gamma\cup\{\Box\}$.  Tiles corresponding to a transition of
the form $\tuple{q,a,b,N,q'}$ are
$\tuple{x,\tuple{q,a},x,\tuple{q',b}}$ and
$\tuple{\tuple{q,a},y,\tuple{q',b},y}$, for all
$x,y\in \Gamma\cup\{\Box\}$.

This ends the definition of the game. It is easy to see that its size
is polynomial in sizes of $M$ and $w$. Now we show that the
constructed game is well-formed.  This is because in any reachable
configuration in any play of the game, the lower edge of the so-far
constructed tiling encodes a configuration of $M$ and thus there are
precisely two (adjacent) occurrences of a color from $Q\times
\Gamma$. Consider the two tiles $t_1=\tuple{\_,\_,\_,\tuple{q,a}}$ and
$t_2=\tuple{\_,\_,\tuple{q,a},\_}$ with these occurrences. If $q$ is
final then the only possibility is to finish the play by completing
the white lower edge, so let us assume that $q$ is not final. Then, by
assumption on $M$, there are precisely two transitions of the form
$\tuple{q,a,\_,\_,\_}$ in $\delta$ and the chosen set contains tiles
corresponding to only one of these transitions. Therefore it contains
precisely one tile $t_1'$ and precisely one $t_2'$ such that
$\tuple{t_1',t_2'}\in H$, $\tuple{t_1,t_1'}\in V$ and
$\tuple{t_2,t_2'}\in V$, so the player has to choose these two
tiles. Then the rest of the row is determined by the chosen two tiles,
and again the lower edge of the constructed tiling correctly encodes a
configuration of $M$; moreover, this configuration is a successor of
the configuration encoded in the previous row.

\renewcommand{\sqtile}[4]{
\begin{tabular}[c]{|@{\hspace{1mm}}c@{\hspace{1mm}}|@{\hspace{1mm}}c@{\hspace{1mm}}|}
  \hline \makebox[2ex]{$#1$}&\makebox[2ex]{$#2$}\\ \hline
  \makebox[2ex]{$#3$}&\makebox[2ex]{$#4$}\\\hline    
\end{tabular}}

\begin{figure}[htb]
\hfill
  \begin{minipage}[b]{3.25cm}
%    \rule{3.1cm}{1mm}\\
\sqtile{x}{y}{x}{{q'\!y}}\hspace{-5pt}
\sqtile{y}{{qa}}{{q'\!y}}{b}\hspace{-5pt}
\sqtile{{qa}}{z}{b}{z}
  \end{minipage}\hfill
  \begin{minipage}[b]{3.25cm}
    % \rule{.75cm}{1mm}\\
     \sqtile{x}{qa}{x}{b}\hspace{-5pt}
\sqtile{qa}{y}{b}{q'\!y}\hspace{-5pt}
\sqtile{y}{z}{q'\!y}{z}
      \end{minipage}\hfill
      \begin{minipage}[b]{2.25cm}
%    \rule{2cm}{1mm}\\
\sqtile{x}{qa}{x}{q'\!b}\hspace{-5pt}
\sqtile{qa}{y}{q'\!b}{y}
      \end{minipage}\hfill~
  \caption{Tiles corresponding to transitions   $\tuple{q,a,b,L,q'}$  $\tuple{q,a,b,R,q'}$  and $\tuple{q,a,b,N,q'}$}
  \label{fig:trans}
\end{figure}

It remains to show that the constructed game is correct. Suppose that
$M$ accepts $w$. Then a winning strategy of the existential player is
to always choose the set containing tiles corresponding to the
transition that leads to acceptance. Conversely, if $M$ does not
accept $w$, then the universal player has a winning strategy: he can
always choose the set containing tiles corresponding to the transition
that does not lead to acceptance.
\end{proof}

The theorem above directly leads to the lower bound on the complexity
of tiling games.
\begin{corollary}\label{cor:gamesHard}
  The problem whether Prover has a winning strategy in a tiling game
  is \twoexptime-hard.
\end{corollary}

\subsubsection{From Tilling Games to \texorpdfstring{\FOTwoModDesc}{FO2MOD}}

Here we show the second step of the lower-bound proof for
\FOTwoModDesc, which is a reduction from tiling games to
\FOTwoModDesc.  We are going to encode strategies for the existential
player as trees. Every complete path in such an encoding corresponds
to a correct tiling for $G$. The nodes on such a~path, read from the
root to the leaf, correspond to tiles in the tiling, read row by row
from left to right. Intuitively, most nodes have just one child
corresponding to the tile placed directly to the right. Nodes
corresponding to tiles in the last column should have one or two
children, depending on whether it is the existential or universal
player's turn. Formally the situation is a bit more complicated,
because we are not able to prevent the tree from having additional
branches with no meaning (for example, a node may have several
children, each of which encodes the same right neighbor).

For a given number $n$ we will be using unary predicates
$B_0,\ldots,B_{n-1}$ for counting modulo $2^n$; the predicate $B_i$
will be responsible of the $i$-th bit of the corresponding number;
$B_0$ is the least significant bit. This is a standard construction
that can be found e.g. in~\cite{Kieronski06} or \cite{BBCKLMW-tocl},
so we do not present the details here.  In the following we will be
using several macros that expand in an obvious way to formulas of
length polynomial in $n$ over predicates $B_0,\ldots,B_{n-1}$. Some of
them are listed below.

\newcommand{\Nr}{\mathit{Nr}}
\begin{tabular}{rl}
  $\Nr(x)=0$ & the number encoded in the node $x$ is 0\\
  $\Nr(x)=2^n{-}1$ & the number encoded in the node $x$ is $2^n-1$\\
  $\Nr(x)=\Nr(y)+1$ & the number in  $x$ is the successor of the number
                    in $y$\\
  $\Nr(x)>\Nr(y)$ & the number in  $x$ is greater than the number
                    in $y$\\
  $\ldots$&\ldots
\end{tabular}

For example, the macro  $\Nr(x)>\Nr(y)$ expands to
$$
\bigvee_{i=0}^n \Big(B_i(x)\wedge\neg B_i(y)\wedge\bigwedge_{j=i+1}^n
(B_j(x)\Leftrightarrow B_j(y))\Big)
$$

\begin{theorem}\label{th:gamesToLogic}
  For all well-formed tiling games $G$ there exists a formula
  $\phi \in \FOTwoModDesc$ of size polynomial in the size of $G$ such
  that the existential player has a winning strategy in $G$ if and
  only if $\phi$ is satisfiable over finite trees.
\end{theorem}

\begin{proof}[Sketch of proof]
Let $G=\tuple{C,T_0,T_1,n,\tuple{t_0,\ldots, t_{n}},\Box, L}$. We
are going to define the formula $\phi$ as a conjunction of several
smaller formulas responsible for different aspects of the encoding.
Most of these aspects are routine; the most interesting is adjacency. 

\subparagraph{Numbering of nodes.} We start by numbering nodes on
paths in the underlying tree. These numbers encode the column numbers
of the corresponding tiles.  The following formulas express that the root is
numbered $0$, the number of any other node is the number of its father
plus one modulo $2^n$, and that all rows are complete (i.e., they have
$2^n$ tiles). 
 \begin{gather*}
  \forall x \big(\neg (\exists y\; y\succv x) \Rightarrow \Nr(x)=0\big)\\
  \forall x \big( \Nr(x)\neq 2^n{-}1 \Rightarrow (\exists y\; x \succv
  y)\wedge \forall y\; (x\succv y \Rightarrow \Nr(y)=\Nr(x)+1)\big)\\
  \forall x \big( \Nr(x)= 2^n{-}1 \Rightarrow \forall y\; (x\succv y
  \Rightarrow \Nr(y)=0)\big)
 \end{gather*}

 \newcommand{\first}{\mathit{First}}
 \newcommand{\last}{\mathit{Last}}
 \newcommand{\tile}{\mathit{tile}}
 \newcommand{\white}{\mathit{white}}
 \newcommand{\setT}{\mathit{setT}}
 \newcommand{\setL}{\mathit{setL}}
 \newcommand{\move}{\mathit{move}}

 \subparagraph{Tiles and colors.} Let $t_0,\ldots, t_m$ be an
 enumeration of all tiles occurring in the game, that is in
 $T_0\cup T_1\cup \{t_0,\ldots,
 t_n,\tuple{\Box,\Box,\Box,\Box}\}$. The predicates
 $\tile_1,\ldots,\tile_m$ correspond to these tiles. For each color
 $c\in C$ we introduce four predicates
 $\pi_1^c,\pi_2^c,\pi_3^c,\pi_4^c$ that will be used to encode the
 four colors of a tile. The formulas below express that each node in
 the underlying tree corresponds to precisely one tile and is colored
 with the four colors of the tile. We assume here that
 $t_i=\tuple{c_i^1,c_i^2,c_i^3,c_i^4}$. We also introduce predicates
 $\setT_0$, $\setT_1$ and $\setL$ corresponding to the sets $T_0$,
 $T_1$ and $L$, respectively.

   \begin{gather*}
     \forall x \Big(\bigvee_{i=0}^m \big(\tile_i(x) \wedge
     \bigwedge_{j\neq i} \neg \tile_j(x)\big)\Big)\\
     \forall x \bigwedge_{i=1}^4\Big(\bigvee_{c\in C}^m \big(\pi_i^c(x) \wedge
     \bigwedge_{c'\neq c} \neg \pi_i^{c'}(x)\big)\Big)\\
     \bigwedge_{i=0}^m \Big( \forall x \big(\tile_i(x)\Leftrightarrow
     \pi_1^{c_i^1}(x)\wedge\pi_2^{c_i^2}(x)\wedge\pi_3^{c_i^3}(x)
     \wedge\pi_4^{c_i^4}(x)\big)\Big)\\
    \Big( (\bigvee_{t_i\in T_0}\tile_i(x))\Leftrightarrow \setT_0(x) \Big)\wedge
     \Big( (\bigvee_{t_i\in T_1}\tile_i(x))\Leftrightarrow \setT_1(x)\Big)\wedge
     \Big( (\bigvee_{t_i\in L}\tile_i(x))\Leftrightarrow \setL(x)\Big)
  \end{gather*}
  
  \subparagraph{First and last row.} The predicates $\first$ and
  $\last$ are used to distinguish the first and the last row of a
  correct tiling. A node $x$ corresponds to a tile in the first row if
  there is no other tile in the same column in previous rows. The last
  row is described dually. The first row is built from tiles
  $t_0,\ldots, t_{n}$ and padded out with white tiles, the last row is
  built only from tiles in $L$. 
   We assume here
 that $\tile_\white$ corresponds to the tile
 $\tuple{\Box,\Box,\Box,\Box}$.

 \begin{gather*}
    \forall x \big(\first(x)\Leftrightarrow \neg (\exists y\; y\lessv
    x)\wedge \Nr(x)=\Nr(y)\big)\\
    \forall x \big(\last(x)\Leftrightarrow \neg (\exists y\; x\lessv
    y)\wedge \Nr(x)=\Nr(y)\big)\\
    \bigwedge_{i=0}^n \Big(\forall x\big(\first(x)\wedge \Nr(x)=i
    \Rightarrow \tile_i(x)\big)\Big)\\
    \forall x\big(\first(x)\wedge \Nr(x)> n
    \Rightarrow \tile_\white(x)\big)\\
     \forall x\big(\first(x)\Rightarrow \pi_1^\Box(x)\wedge\pi_2^\Box(x) )\\
    \forall x\; \big(\last(x)
    \Rightarrow  \pi_3^\Box(x)\wedge\pi_4^\Box(x) \wedge\setL(x)\big)
  \end{gather*}

  \subparagraph{Existential and universal rows.} Each element in each
  row is marked with predicate $E$ or $A$ depending on which player's
  turn it is. Each row is marked the same (each element has the same
  marking as its left neighbor, if it exists). The first row is
  existential and then the marking alternates between existential and
  universal.

  \begin{gather*}
    \forall x\big( E(x)\Leftrightarrow \neg A(x)\big)\\
    \forall x \big( \Nr(x)\neq 0 \Rightarrow \exists y\; (y\succv x
    \wedge (E(y)\Leftrightarrow E(x))\big)\\
    \forall x (\first(x)\Rightarrow E(x))\\
    \forall x \big( \Nr(x)= 0 \wedge \neg\first(x)\Rightarrow \exists
    y\; (y\succv x \wedge (E(y)\Leftrightarrow A(x))\big)
  \end{gather*}

  \subparagraph{Universal rows have two successors.}  Each non-first
  row is marked with predicate $\move_0$ or $\move_1$ depending on the
  set of tiles ($T_0$ or $T_1$) from which it is built. Universal
  non-last rows have two successors, marked respectively with
  $\move_0$ and $\move_1$.

  \begin{gather*}
    \forall x\; \Big(\neg\first(x)\Rightarrow
    \big(\move_0(x)\Leftrightarrow \neg \move_1(x)\big)\Big)\\
    \forall x \big( \Nr(x)\neq 0 \Rightarrow \exists y\; (y\succv x
    \wedge (\move_0(y)\Leftrightarrow \move_0(x))\big)\\
    \forall x\; (\move_0(x)\Rightarrow \setT_0(x))\wedge \forall x\;
    (\move_1(x)\Rightarrow \setT_1(x))\\
    \forall x \Big( \Nr(x)= 2^n{-}1 \wedge \neg\last(x)\wedge
    A(x)\Rightarrow \big(\exists y\; (x\succv y \wedge
    \move_0(y)\big)\wedge \exists
    y\; (x\succv y \wedge \move_1(y))\big)\Big)\\
  \end{gather*}

 \subparagraph{Horizontal adjacency.} This is simple. We first
 establish the white frame on the first and last tile in each row and
 then simply say that for each non-first tile in any row the left edge
 of the tile matches the right edge of the preceding tile. 
 \begin{gather*}
   \forall x\; \big(\Nr(x)=0\Rightarrow \pi_1^\Box(x)\wedge\pi_3^\Box(x)\big)\\
   \forall x\; \big(\Nr(x)=2^n{-}1\Rightarrow
   \pi_2^\Box(x)\wedge\pi_4^\Box(x)\big)\\
   \bigwedge_{c\in C}\Big(\forall x\big((\Nr(x)\neq 0) \wedge \pi_1^c(x)\Rightarrow \exists
   y\; (y\succv x \wedge \pi_2^c(y))\big)\Big)\\
   \bigwedge_{c\in C}\Big(\forall x\big((\Nr(x)\neq 0) \wedge\pi_3^c(x)\Rightarrow \exists
   y\; (y\succv x \wedge  \pi_4^c(y))\big)\Big)\\
 \end{gather*}

 \subparagraph{Vertical adjacency.} This is the tricky part of the
 encoding. The difficulty comes from the fact that we cannot number
 rows of the tiling and we have no means to say that two tiles occur
 in consecutive rows. Using predicates $B_0,\ldots, B_{n-1}$ we can
 number the $2^n$ columns, but the number of rows may be much higher
 and we cannot afford having enough predicates for numbering
 them. Therefore, we can say that two tiles occur in the same column
 (or in consecutive columns), but we cannot do the same with rows. On
 the other hand, when we add a new tile to a row, we have to make sure
 that the upper edge of the new tile matches the lower edge of the
 tile directly above, so we have to read the colors of the tile
 directly above. For non-white colors this can be done by observing
 that each occurrence of a color in an upper edge of a row must be
 matched with another occurrence of the same color in a lower edge of
 the previous row in the same column, therefore the number of
 occurrences of each color in each column should be even. Hence the
 color directly above the upper edge of the current row is the only
 non-white color that occurs an odd number of times in the current
 column in preceding rows. In the case of white color we have to take
 into consideration the upper white edge of the constructed tiling, so
 the color directly above is white if it occurs an even number of
 times in the current column in preceding rows.

 \begin{gather*}
   \forall x\big(\pi_1^\Box(x)\Rightarrow \existsmod{=}{0}{2}y\;\;
   y\lessv x \wedge  \Nr(y){=}\Nr(x)\wedge (\pi_1^\Box(y)\vee\pi_3^\Box(y))\big)\\
   \forall x\big(\pi_2^\Box(x)\Rightarrow \existsmod{=}{0}{2}y\;\;
   y\lessv x \wedge  \Nr(y){=}\Nr(x)\wedge (\pi_2^\Box(y)\vee\pi_4^\Box(y))\big)\\
   \bigwedge_{c\in C\setminus\{\Box\}}\Big(\forall
   x\big(\pi_1^c(x)\Rightarrow \existsmod{=}{1}{2}y\; y\lessv x \wedge
\Nr(y){=}\Nr(x)\wedge   (\pi_1^c(y)\vee\pi_3^c(y))\big)\Big)\\
   \bigwedge_{c\in C\setminus\{\Box\}}\Big(\forall
   x\big(\pi_2^c(x)\Rightarrow \existsmod{=}{1}{2}y\; y\lessv x \wedge
\Nr(y){=}\Nr(x)\wedge   (\pi_2^c(y)\vee\pi_4^c(y))\big)\Big)\\
 \end{gather*}

 \subparagraph{Correctness of the constructed formula.} Let $\phi$ be
 the conjunction of all formulas mentioned above. It is not
 difficult to see that the size of $\phi$ is polynomial in the size of
 $G$: the longest part of $\phi$ is the encoding of tiles and colors,
 which is proportional in length to the sum of squared numbers of
 tiles and colors.

 Assume that $\phi$ has a finite model $\str{M}$. The formula
 \emph{Tiles and colors} guarantees that each node in $\str{M}$
 directly encodes precisely one tile. The formula \emph{Numbering of
   nodes} guarantees that nodes on each root-to-leaf path are
 correctly numbered modulo $2^n$, with the root numbered $0$ and the
 leaf numbered $2^n-1$. Thus each segment of length $2^n$ of such
 a~path, consisting of nodes numbered from $0$ to $2^n-1$, corresponds
 to one row of tiles. The \emph{Horizontal adjacency} formula
 guarantees that (horizontally) adjacent tiles in such a row have
 matching colors. Similarly, \emph{Vertical adjacency} formula
 guarantees that tiles occurring on the same position in two
 consecutive rows (that is, vertically adjacent tiles) have matching
 colors. The \emph{First and last row} formula guarantees that the
 first row is initial and the last row is built from tiles in
 $L$. Therefore each complete path in $\str M$ encodes a correctly
 tiled corridor. By the \emph{Existential and universal rows} formula
 all rows encoded in any such path are alternately marked as
 existential and universal, starting with an existential one. Finally,
 the \emph{Universal rows have two successors} formula guarantees that
 each encoding of a universal row in $\str M$ is followed by encodings
 of two existential rows, one build from tiles in $T_0$ and one from
 tiles in $T_1$. Thus (here we use the assumption that the game is
 well-formed) $\str M$ covers both possible moves of the universal
 player. Now the strategy of the existential player is to follow the
 path in $\str M$ corresponding to the partial tiling as it is
 constructed during the game. She starts in the root of $\str M$.
 Then, every time when it is her turn, the existential player reads
 from $\str M$ any successor row of her current position and replies
 with this row, moving down the tree to the position of the successor
 row.  Every time when it is the universal player's turn, his both
 possible moves are encoded as successor rows of the current position
 of the existential player, so she can always follow the branch
 chosen by the universal player. Since $\str M$ is finite, each play
 stops after finitely many rounds with the constructed tiling
 corresponding to a path in $\str M$. Therefore in each play we obtain
 a correctly tiled corridor and the existential player wins.

 For the other direction, assume that the existential player has a
 winning strategy. We construct a model $\str M$ of $\phi$
 inductively, level by level, as follows.  We start with $2^n$ nodes,
 number them from $0$ to $2^n-1$ using predicates
 $B_0,\ldots, B_{n-1}$ and connect consecutive nodes with predicate
 $\succv$. Then we label nodes numbered $0$ to $n$ with predicates
 $\tile_0,\ldots,\tile_n$, respectively, and nodes numbered $n+1$ to
 $2^n-1$ with $\tuple{\Box,\Box,\Box,\Box}$. Then (and later, always
 after adding new labels of the form $\tile_i$) we add appropriate
 labels $\pi^c_j, \setT_0, \setT_1$ and $\setL$ to make the formula
 \emph{Tiles and colors} true. Similarly, we add labels $\lessv$ to
 all pairs of nodes that are connected by the transitive closure of
 $\succv$.  We also label all these nodes with predicates $\first$ and
 $E$. This way we obtain the encoding of the initial row in the
 game. Next we inductively, level by level, construct the remaining
 parts of $\str M$.

 Assume that $\str M$ is constructed up to some level $k$, with all
 leaves labeled $E$, and that each path constructed so far encodes a
 partial tiling in some play after $k$ rounds, with existential
 player's turn. We extend $\str M$ leaf by leaf. Let us consider one
 such leaf $\ell$, it encodes the last tile in some row. If all tiles
 in this row are in the set $L$, we label all nodes in this row with
 predicate $\last$ and finish the construction. Otherwise let $r$ be
 the next row given by the winning strategy of the existential player
 in the partial play encoded by the path from the root to $\ell$. We
 take fresh $2^n$ nodes and extend the path from root to $\ell$ with
 the encoding of a row $r$, as above, but this time labeling all new
 nodes with predicate $A$ indicating the universal player's move. We
 also label all new nodes with predicate $\move_0$ or $\move_1$,
 depending on whether $r$ is constructed from tiles in $T_0$ or $T_1$,
 respectively.  Again, if the row is final, it is marked with
 predicate $\last$, otherwise we take twice $2^n$ fresh nodes and
 connect to the current leaf two segments of length $2^n$ that encode
 two possible moves of the universal player. We mark each node in both
 segments with predicate $E$ and thus finish the construction of level
 $k+1$ at node $\ell$. The construction is repeated for all leafs at
 level $k$. Since during the construction we follow the winning
 strategy, no encoded play can last forever and the construction ends
 after finitely many iterations. By inspection of all conjuncts one
 can check that the constructed tree is a model of $\phi$.
\end{proof}

As a corollary of the theorem above and Corollary~\ref{cor:gamesHard}
we get the main theorem of this section.
\begin{theorem}
\label{th:lower}
The satisfiability problem for \FOTwoModDesc interpreted over finite
trees is \twoexptime-hard.
\end{theorem}

\section{Conclusions and future work}
We have shown that the satisfiability problem for two-variable logic
extended with modulo counting quantifiers and interpreted over finite
trees is \twoexptime-complete. The upper bound is based on the
small-model property of the logic; for the lower bound we have
developed a version of tiling games for \aexpspace computations.

There are several possible directions for future work. One of them is
studying restrictions or extensions of the logic presented
here. Natural candidates for restrictions include \emph{guarded
  fragment} of the logic or \emph{unary alphabet restriction} as
in~\cite{BBCKLMW-tocl}; natural extensions include arbitrary
uninterpreted binary symbols as in~\cite{BednarczykChK-CSL17}. Another
possibility is investigation of the (finite) satisfiability problem
for \FOTwoMod\ on arbitrary structures---we even do not know whether
this problem is decidable. Yet another direction is to study the
expressive power of the logic and to find an expressively equivalent
extension of CTL.

\bibliography{bibliography}

\endgroup

\end{document}